\documentclass[12pt]{article}
\usepackage{mydef}
\usepackage[hide]{ed}
\usepackage{appendix}
\begin{document}
\title{Splitting and Matrix Exponential approach for jump-diffusion models with Inverse Normal Gaussian, Hyperbolic and Meixner jumps}

\author{Andrey Itkin}

\affil{\small New York University, School of Engineering,\\
6 Metro Tech Center, RH 517E, Brooklyn NY 11201, USA \\
aitkin@nyu.edu}

\date{\today}

\maketitle

\begin{abstract}
This paper is a further extension of the method proposed in \cite{Itkin2014} as applied to another set of jump-diffusion models: Inverse Normal Gaussian, Hyperbolic and Meixner. 
To solve the corresponding PIDEs we accomplish few steps. First, a second-order operator splitting on financial processes (diffusion and jumps) is applied to these PIDEs. To solve the diffusion equation, we use standard finite-difference methods. 
For the jump part, we transform the jump integral into a pseudo-differential operator and construct its second order approximation on a grid which supersets the grid that we used for the diffusion part. The proposed schemes are unconditionally stable in time and preserve positivity of the solution which is computed either via a matrix exponential, or via P{\'a}de approximation of the matrix exponent. Various numerical experiments are provided to justify these results.

\end{abstract}

%%%%%%%%%%%%%%%%%%%%%%%%%%%%%%%%%%%%%%%%%%%%%%%%%
\section{Introduction}
%%%%%%%%%%%%%%%%%%%%%%%%%%%%%%%%%%%%%%%%%%%%%%%%%
This paper is a further extension of the method proposed in \cite{Itkin2014} as applied to another set of jump-diffusion models: Inverse Normal Gaussian, Hyperbolic and Meixner. These models have been introduced in mathematical finance within last two decades, see \cite{NIG,EberleinKeller:95,EberleinKellerPrause:98,Schoutens01}. However, to the best of author's knowledge it seems they received less attention of practitioners as compared, e.g., with Merton, Kou and CGMY/KoBoL models, see \cite{Itkin2014} for a short survey of the latter models and references therein. At the first glance this looks unfair, because typical distributions of returns produced by the former models fit the observed market data (and, in particular, fat tails and skew) even better than their more popular counterparts. One of the possible reasons could be that despite the pdf and the characteristic function (CF) of these models are known in closed form, and, therefore, pricing of plain vanilla and even American options could be done transform methods (FFT, cosine, adaptive integration in the Fourier space), see \cite{CarrMadan:99a,Lewis:2001, Lipton2001, FangOosterlee2008,Lord2007}, the analytic expressions for the pdf and CF are more complicated than that for their counterparts, and sometimes require usage of special functions. However, the latter doesn't prevent pricing and hedging of simple vanilla instruments to be efficiently done.

The second point is that the considered models are pure jump models that don't contain a diffusion component. This, however, could be easy relaxed.

On the other hand, nowadays practitioners want to have a model which is capable to simultaneously fit market data on both vanilla and exotic options. For doing that they need to consider kind of a stochastic local volatility (LSV) model, or even a LSV model with jumps. Under these conditions neither the pdf, nor the CF are available in closed form. Therefore, efficient numerical methods should be used to solve the pricing equations which belong to the class of Partial Integro-Differential Equations (PIDE).

A number of methods were proposed to address such a construction, see \cite{Itkin2014} and references therein as well as discussion of problems related to their implementation. In particular, they include a discretization of the PIDE that is implicit in the differential terms and explicit in the integral term (\cite{ContVolchkova2003}), Picard iterations for computing the integral equation (\cite{Halluin2004, Halluin2005a}) and a second-order accurate, unconditionally-stable operator splitting (ADI) method that does not require an iterative solution of an algebraic equation at each time step (\cite{AA2000}). Various forms of operator splitting technique were also used for this purpose (\cite{ItkinCarr2012Kinky, Itkin2014}). In this paper, we will review operator splitting on financial processes in more detail.

Assuming that an efficient discretization of the PIDE in time was properly chosen, the remaining problem is a fast computation of the jump integral, as it was observed to be relatively expensive. A short survey of various proposed method including review of their advantages and disadvantages again could be found in \cite{Itkin2014}. Also in that paper the original method proposed in \cite{ItkinCarr2012Kinky} was further elaborated on exploiting the following idea. First, we use an operator-splitting method on the financial processes, thus separating the computation of the diffusion part from the integral part. Then, similar to \cite{ItkinCarr2012Kinky}, we represent the jump integral in the form of a pseudo-differential operator. Next we formally solve the obtained evolutionary partial pseudo-differential equations via a matrix exponential.

In \cite{Itkin2014} it was shown that the matrix exponential can be efficiently computed for Merton's, Kou's and CGMY models, and that the efficiency of this method is not worse than that of the FFT. It was also mentioned that the proposed method is almost universal, i.e., it allows computation of PIDEs for various jump-diffusion models in a unified form. It is also important that this method is relatively simple for implementation.

In the present paper we apply this approach to the Inverse Normal Gaussian, Hyperbolic and Meixner models. We construct finite difference schemes to solve the corresponding PIDEs and prove that all the proposed schemes demonstrate the second order convergence in space and time and are unconditionally stable. The suggested approach is new and eliminates some known limitations of the existing methods, see discussion in \cite{Itkin2014}. Also, for the first time splitting and matrix exponential method is used as applied to the referenced jump models. This allows an efficient usage of these models in a more complex framework, e.g., the LSV model with jumps. Furthermore, the complexity of solving a pure jump evolutionary equation for the Meixner model using the new method is close to $O(N)$ which is better than that of the FFT.

Finally, as the distributions underlying the corresponding L\'evy processes are capable to fit well the market data, using these jump models together with the efficient numerical method for solving the jump-diffusion PIDEs potentially gives rise to a more efficient pricing and hedging of the derivative instruments.

Also to underline, despite the general idea of the method  has been already described in \cite{ItkinCarr2012Kinky, Itkin2014} constructing a particular discretization of the jump operators could be different for every model. This requires the corresponding Propositions to be proved in every case to demonstrate approximation, stability and complexity of the method. Therefore, these schemes, Propositions and proofs are the main contributions  of this paper.

The rest of the paper is organized as follows. In section~\ref{Sec2} we briefly discuss a general form of a backward PIDE for the class of L{\'e}vy models. In Section~\ref{Sec4} we present our general approach to the solution of the PIDE using splitting and the matrix exponential approach. An explicit construction of various finite-difference schemes of the first and second order is presented in the next section. The results presented in that section are new, and to the best of our knowledge have not been discussed in the literature. Our technique utilizes some results from matrix analysis related to definitions of M-matrices, Metzler matrices and eventually exponentially nonnegative matrices. We also provide the results of various numerical tests to demonstrate convergence of our method. The final section concludes.

%%%%%%%%%%%%%%%%%%%%%%%%%%%%%%%%%%%%%%%%%%%%%%%%%%%%%%%%%%%%
\section{L{\'e}vy Models and Backward PIDE } \label{Sec2}
%%%%%%%%%%%%%%%%%%%%%%%%%%%%%%%%%%%%%%%%%%%%%%%%%%%%%%%%%%%%
To avoid uncertainty, let us consider the problem of pricing equity options written on a single stock. As we will see, this specification does not cause us to lose any generality, but it makes the description more practical. We assume an underlying asset (stock) price $S_t$ be driven by an exponential of a L{\'e}vy process
\begin{equation} \label{Levy}
S_t = S_0 \exp (L_t), \quad 0 \le t \le T,
\end{equation}
where $t$ is time, $T$ is option expiration, $S_0 = S_t \ |_{t=0}$, $L_t$ is the L{\'e}vy process $L = (L_t)_{0 \le t \le T}$ with a nonzero Brownian (diffusion) part. Under the pricing measure, $L_t$ is given by
\begin{equation} \label{Lt}
L_t = \gamma t  + \sigma W_t + Y_t, \qquad \gamma, \sigma \in \mathbb{R}, \quad \sigma > 0,
\end{equation}
with  \LY triplet $(\gamma, \sigma, \nu$), where $W_t$ is a standard Brownian motion on $0 \le t \le T$ and $Y_t$ is a pure jump process.
We consider this process under the pricing measure, and therefore $e^{-(r-q) t} S_t$  is a martingale, where $r$ is the interest rate and $q$ is a continuous dividend. This allows us to express $\gamma$ as (\cite{Eberlein2009})
\[
\gamma = r - q - \frac{\sigma^2}{2} - \int_\mathbb{R} \left(e^x -1 -x {\bf 1}_{|x| < 1}\right)\nu(dx),
\]
\noindent where $\nu(dx)$ is a \LY measure which satisfies
\[ \int_{|x| > 1}e^x \nu(dx) < \infty.  \]

We leave $\nu(dx)$ unspecified at this time, because we are open to consider all types of jumps including those with finite and infinite variation, and finite and infinite activity.
\footnote{We recall that a standard Brownian motion already has paths of infinite variation. Therefore, the \LY process in \eqref{Lt} has infinite variation since it contains a continuous martingale component. However, here we refer to the infinite variation that comes from the jumps.}

To price options written on the underlying process $S_t$, we want to derive a PIDE that describes time evolution of the European option prices $C(x,t), \ x \equiv \log (S_t/S_0)$. Using a standard martingale approach, or by creating a self-financing portfolio, one can derive the corresponding PIDE (\cite{ContTankov})
\begin{multline} \label{PIDE}
r C(x,t) = \fp{C(x,t)}{t} + \left(r-\frac{1}{2}\sigma^2 \right) \fp{C(x,t)}{x} + \frac{1}{2}\sigma^2 \sop{C(x,t)}{x} \\
+  \int_\mathbb{R}\left[ C(x+y,t) - C(x,t) - (e^y-1)\fp{C(x,t)}{x} \right] \nu(dy)
\end{multline}
for all $(x,t) \in \mathbb{R} \times (0,T)$, subject to the terminal condition
\begin{equation}
C(x,T) = h(x),
\end{equation}
where $h(x)$ is the option payoff, and some boundary conditions which depend on the type of the option. The solutions of this PIDE usually belong to the class of viscosity solutions (\cite{ContTankov}).

We now rewrite the integral term using the following idea. It is well known from quantum mechanics (\cite{OMQM}) that a translation (shift) operator in $L_2$ space could be represented as
\begin{equation} \label{transform}
    \mathcal{T}_b = \exp \left( b \dfrac{\partial}{\partial x} \right),
\end{equation}
with $b$ = const, so
\[ \mathcal{T}_b f(x) = f(x+b). \]

Therefore, the integral in \eqref{PIDE} can be formally rewritten as
\begin{align} \label{intGen}
\int_\mathbb{R} \left[ C(x+y,t) \right. & \left. - C(x,t) - (e^y-1) \fp{C(x,t)}{x} \right] \nu(dy) =  \mathcal{J} C(x,t), \\
\mathcal{J} & \equiv \int_\mathbb{R}\left[
\exp \left( y \dfrac{\partial}{\partial x} \right) - 1 - (e^y-1) \fp{}{x} \right] \nu(dy). \nonumber
\end{align}

In the definition of operator $\mathcal{J}$ (which is actually an infinitesimal generator of the jump process), the integral can be formally computed under some mild assumptions about existence and convergence if one treats the term $\partial/ \partial x$ as a constant. Therefore, operator $\mathcal{J}$ can be considered as some generalized function of the differential operator $\partial_x$. We can also treat $\mathcal{J}$ as a pseudo-differential operator.

For the future, an important observation is that
\begin{equation} \label{MGT}
\mathcal{J} = \phi(-i \partial_x)
\end{equation}
\noindent where $\phi(u)$ is the characteristic exponent of the jump process. This directly follows from the L{\'e}vy-Khinchine theorem. Note, that thus defined characteristic exponent is computed using the expectation under a risk-neutral measure. In other words, the last term in the definition of $\mathcal{J}$ is a compensator which is added to make the forward price to be a true martingale under this measure.

With allowance for this representation, the whole PIDE in the \eqref{PIDE} can be rewritten in operator form as
\begin{equation} \label{oper}
\partial_\tau C(x,\tau) = [\mathcal{D} + \mathcal{J}]C(x,\tau),
\end{equation}
\noindent where $\tau = T - t$ and $\mathcal{D}$ represents a differential (parabolic) operator
\begin{equation}
\mathcal{D} \equiv - r + \left(r-\frac{1}{2}\sigma^2 \right) \fp{}{x} + \frac{1}{2}\sigma^2 \sop{}{x},
\end{equation}
 where the operator $\mathcal{D}$ is an infinitesimal generator of diffusion.

Notice that for jumps with finite variation and finite activity, the last two terms in the definition of the jump integral $\mathcal{J}$ in \eqref{PIDE} could be integrated out and added to the definition of $\mathcal{D}$. In the case of jumps with finite variation and infinite activity, the last term could be integrated out. However, here we will leave these terms under the integral for two reasons: i) this transformation (moving some terms under the integral to the diffusion operator) does not affect our method of computation of the integral, and ii) adding these terms to the operator $\mathcal{D}$ could potentially negatively influence the stability of the finite-difference scheme used to solve the parabolic equation
$\mathcal{D} C(x,t) = 0$. This equation naturally appears as a part of our splitting method, which is discussed in the next section.

%%%%%%%%%%%%%%%%%%%%%%%%%%%%%%%%%%%%%%%%%%%%%%%%%%%%%%%%%%%%
\section{Operator Splitting and Solution of Jump Equations} \label{Sec4}
%%%%%%%%%%%%%%%%%%%%%%%%%%%%%%%%%%%%%%%%%%%%%%%%%%%%%%%%%%%%
To solve \eqref{oper} we use splitting. This technique is also known as the method of fractional steps (see \cite{yanenko1971, samarski1964, dyakonov1964}) and sometimes is cited in financial literature as Russian splitting or locally one-dimensional schemes (LOD) (\cite{Duffy}).

Below we follow \cite{Itkin2014} to give a short survey of this technique as applied to linear and nonlinear PDEs.

The method of fractional steps reduces the solution of the original $k$-dimensional unsteady problem to the solution of $k$ one-dimensional equations per time step. For example, consider a two-dimensional diffusion equation with a solution obtained by using some finite-difference method. At every time step, a standard discretization on space variables is applied, such that the finite-difference grid contains $N_1$ nodes in the first dimension and $N_2$ nodes in the second dimension. Then the problem is solving a system of $N_1 \times N_2$ linear equations, and the matrix of this system is block-diagonal. In contrast, utilization of splitting results in, e.g.,  $N_1$ systems of $N_2$ linear equations, where the matrix of each system is banded (tridiagonal). The latter approach is easy to implement and, more importantly, provides significantly better performance.

The previous procedure uses operator splitting in different dimensions. \cite{marchuk1975} and then \cite{Strang1968} extended this idea for complex physical processes (for instance, diffusion in the chemically reacting gas, or the advection-diffusion problem). In addition to (or instead of) splitting on spatial coordinates, they also proposed splitting the equation into physical processes that differ in nature, for instance, convection and diffusion. This idea becomes especially efficient if the characteristic times of evolution (relaxation time) of such processes are significantly different.

For the PIDE in \eqref{PIDE} we use a version of splitting described in \cite{Itkin2014}
which gives rise to the following numerical scheme:
\begin{align} \label{splitFin}
C^{(1)}(x,\tau) &= e^{\frac{\Delta \tau}{2} \mathcal{D} } C(x,\tau), \\
C^{(2)}(x,\tau) &= e^{\Delta \tau \mathcal{J}} C^{(1)}(x,\tau) \nonumber, \\
C(x,\tau+\Delta \tau) &= e^{\frac{\Delta \tau}{2} \mathcal{D} } C^{(2)}(x,\tau).  \nonumber
\end{align}

Thus, instead of an unsteady PIDE, we obtain one PIDE with no drift and diffusion (the second equation in \eqref{splitFin}) and two unsteady PDEs (the first and third ones in \eqref{splitFin}).

In what follows, we consider how to efficiently solve the second equation, while assuming that the solution of the first and the third equations can be obtained using any finite-difference method that is sufficiently efficient. To this end, in various examples given in the next sections we will explicitly mention what particular method was used for this purpose.

To solve the second (pure jump) evolutionary equation we again follow the method of \cite{Itkin2014}. By definition of the jump generator $\mathcal{J}$, under some mild constraints on its existence, $\mathcal{J}$ could be viewed as a function of the operator $\partial_x$. Therefore, solving the integral (second) equation in \eqref{splitFin} requires a few steps.

First, we construct an appropriate discrete grid ${\bf G}(x)$ in the truncated (originally infinite) space domain. This grid could be nonuniform. An important point is that in the space domain where the parabolic equations of \eqref{splitFin} are defined, this grid should coincide with the finite-difference grid constructed for the solution of these parabolic equations. In other words, the PIDE grid is a superset of the PDE grid. This is useful to avoid interpolation of the solution that is obtained on the jump grid (the second step of the splitting algorithm) to the diffusion grid that is constructed to obtain solutions at the first and third splitting steps.

For the sake of concreteness let the parabolic equation be solved at the space domain $[x_0,x_k],$ $\ x_0 > -\infty$, $x_k < \infty$ using a nonuniform grid with $k+1$ nodes ($x_0, x_1,...,x_k$) and space steps $h_1 = x_1-x_0, ..., h_k = x_k - x_{k-1}$. The particular choice of $x_0$ and $x_k$ is determined by the problem under consideration. We certainly want $|x_0|$ and  $|x_k|$ not to be too large.  The integration limits of $\mathcal{J}$ in \eqref{intGen} are, however,  plus and minus infinity. Truncation of these limits usually is done to fit memory and performance requirements. On the other hand, we want a fine grid close to the option strike $K$ for better accuracy. Therefore, a reasonable way to construct a jump grid is as follows. For $x_0 \le x \le x_k$, the jump grid coincides with the grid used for solution of the parabolic PDEs. Outside of this domain, the grid is expanded by adding nonuniform steps; i.e., the entire jump grid is $x_{-K}, x_{1-K}, ... x_{-1}, x_0, x_1, ..., x_k, x_{k+1}, ..., x_{k+M}$. Here $K >0, \ M>0$ are some integer numbers that are chosen based on our preferences. Since contribution to $\mathcal{J}$  from very large values of $x$ is negligible, the outer grid points $x_{-K}, x_{1-K}, ... x_{-1}$ and $x_{k+1}, ..., x_{k+M}$ can be made highly nonuniform. One possible algorithm could be to have the steps of these grids be a geometric progression. This allows one to cover the truncated infinite interval with a reasonably small number of nodes.

Second, the discretization of $\partial_x$ should be chosen on ${\bf G}(x)$. We want this discretization to:
\begin{enumerate}
\item Provide the necessary order of approximation of the whole operator $\mathcal{J}$ in space.
\item Provide unconditional stability of the solution of the second equation in \eqref{splitFin}.
\item Provide positivity of the solution.
\end{enumerate}

Let  $\Delta_x$ denote a discrete analog of $\partial_x$ obtained by discretization of $\partial_x$ on the grid ${\bf G}(x)$. Accordingly, let us define the matrix $J(\Delta_x)$ to be the discrete analog of the operator $\mathcal{J}$ on the grid ${\bf G}(x)$. In \cite{Itkin2014}
the following proposition is proven that translates the above requirements to the conditions on $J(\Delta_x)$.
\begin{proposition} \label{prop0}
The finite-difference scheme
\begin{equation} \label{fd0}
C(x, \tau + \Delta \tau) = e^{\Delta \tau J(\Delta_x)} C(x,\tau)
\end{equation}
is unconditionally stable in time $\tau$ and preserves positivity of the vector $C(x,\tau)$ if there exists an M-matrix $B$ such that $J(\Delta_x) = - B$.
\end{proposition}

This proposition gives us a recipe for the construction of the appropriate discretization of the operator $\mathcal{J}$. In the next section, we will give some explicit examples of this approach.

Once the discretization is performed, all we need is  to compute a matrix exponential $ e^{\Delta \tau J(\Delta_x)}$, and then a product of this exponential with $C(x,\tau)$. The following facts make this method competitive with those briefly described in the introduction. We need to take into account that:
\begin{enumerate}
\item The matrix $J(\Delta_x)$ can be precomputed once the finite-difference grid ${\bf G}(x)$ has been built.

\item If a constant time step is used for computations, the  matrix $\mathcal{A} = e^{\Delta \tau J(\Delta_x)}$ can also be precomputed.
\end{enumerate}

If the above two statements are true, the second splitting step results in computing a product of a matrix with time-independent entries and a vector. The complexity of this operation is $O(N^2)$, assuming the matrix $\mathcal{A}$ is $N \times N$, and the vector is $N \times 1$.
However, $N$ in this case is relatively small (see some numerical examples and estimates in \cite{Itkin2014}). Also the product $\mathcal{A} C(x,\tau)$ can be computed using FFT, if at every time step one re-interpolates values from ${\bf G}(x)$ to the FFT grid, similar to how this was done in \cite{Halluin2004}. This reduces the total complexity to $O(N \log_2 N)$.

In some special cases (Merton's jump model, Kou model) the product $\mathcal{A} C(x,\tau)$ could be computed with the complexity $O(N)$ using some tricks proposed in \cite{Itkin2014}. We will further exploit this idea for some models described in this paper.

The above consideration is sufficiently general in the sense that it covers any particular jump model where jumps are modeled as an exponential \LY process. Below
we review three jump models: Inverse Normal Gaussian, Hyperbolic and Meixner. Given a model, our goal is to construct a finite-difference scheme, first for $\Delta_x$, and then for $J(\Delta_x)$, that satisfies the conditions of Proposition~\ref{prop0}. Again we want to emphasize that we discuss these jump models being a part of a more general either LV or LSV model with jumps. Otherwise, as the CF of the considered models are known in closed form, any FFT based method would be more efficient in, e.g., obtaining prices of European vanilla options.

%%%%%%%%%%%%%%%%%%%%%%%%%%%%%%%%%%%%%%%%%%%%%%%%%%%%%%%%%%%%
\subsection{Normal Inverse Gaussian Model (NIG)} \label{NIGsect}

The NIG type L\'evy process was introduced by \cite{NIG} as a model for log returns of stock prices. It is a sub-class of the more general class of hyperbolic L\'evy processes that we will consider in the next section. \cite{NIG} considered classes of normal variance-mean
mixtures and defined the NIG distribution as the case when the mixing distribution is inverse Gaussian with the characteristic exponent

\begin{equation} \label{NIG_CE}
\phi_{NIG}(\alpha, \beta,\delta,u) = \delta \left( \sqrt{\alpha^2 - \beta^2} - \sqrt{\alpha^2 - (\beta+ i u)^2} \right).
\end{equation}
Therefore, the CF function reads
\begin{equation} \label{NIG_CF}
\Phi_{NIG}(\alpha, \beta,\delta,\mu,u) = \exp\left\{ t \delta \left( \sqrt{\alpha^2 - \beta^2} - \sqrt{\alpha^2 - (\beta+ i u)^2} \right) + i t u \mu \right\},
\end{equation}
\noindent where $u \in \mathbb{R}, \ \mu \in \mathbb{R} \ \delta > 0, \ 0 \le |\beta| \le \alpha$. It is known that parameters of the NIG play the following role for the underlying distribution: $\alpha$ is responsible for the tail heaviness of steepness, $\beta$ affects the symmetry, $\delta$ scales the distribution, and $\mu$ determines the mean value (location). It is also known that when using the NIG process for option pricing the location parameter of the distribution has no effect on the option value, so one can use $\mu = 0$. This, however, is not critical for the below approach and could be easily relaxed.

The linearity of the log of the characteristic function with respect to time shows that it is an infinitely divisible process with stationary independent increments.

The NIG process has no diffusion component making it a pure jump process with the L\'evy density
\begin{equation} \label{NIGLevydens}
\nu(dx)  = \dfrac{2 \alpha \delta}{\pi} \dfrac{\exp(\beta x) K_1(\alpha |x|)}{|x|} dx,
\end{equation}
\noindent where $K_\lambda(z)$ is the modified Bessel function of the third kind\footnote{They are also known as the modified Bessel functions of the second kind, or Macdonald functions, see \cite{BesselK}.}

Next we need to substitute \eqref{NIGLevydens} into the definition of the operator $\mathcal{J}$ in \eqref{intGen} and fulfill a formal integration to obtain the corresponding evolutionary pure jump equation in the explicit form. However, as it was mentioned earlier this step could be formalized by making use of \eqref{MGT}. Therefore, we immediately obtain
\begin{equation} \label{JING}
\mathcal{J} = \delta \left( \sqrt{\alpha^2 - \beta^2} - \sqrt{\alpha^2 - (\beta+ \triangledown)^2} \right),
\end{equation}
\noindent where $\triangledown \equiv \partial /\partial x$. The corresponding evolutionary pure jump equation to be solved is
\begin{equation} \label{EqMerton}
C^{(2)}(x,\tau) = \mathcal{A} C^{(1)}(x,\tau), \qquad
\mathcal{A} = \exp \left\{ \Delta \tau \delta \left( \sqrt{\alpha^2 - \beta^2} - \sqrt{\alpha^2 - (\beta+ \triangledown)^2} \right) \right\}.
\end{equation}
Before constructing a finite difference scheme to solve this equation
we need to introduce some definitions. Define a one-sided {\it forward} discretization of $\triangledown$, which we denote as $A^F: A^F C(x) = [C(x+h,t) - C(x,t)]/h$. Also define a one-sided {\it backward} discretization of $\triangledown$, denoted as $A^B: \ A^B C(x) = [C(x,t) - C(x-h,t)]/h$. These discretizations provide first order approximation in $h$, e.g., $\triangledown C(x) = A^F C(x) +  O(h)$.
To provide the second order approximations, use the following definitions. Define $A^C_2 = A^F\dot A^B$ - the {\it central} difference approximation of the second derivative $\triangledown^2$, and $A^C = (A^F + A^B)/2$ - the {\it central} difference approximation of the first derivative $\triangledown$. Also define a one-sided second order approximations to the first derivatives: {\it backward} approximation  $A^B_2: \ A^B_2 C(x) = [ 3 C(x) - 4 C(x-h) + C(x-2h)]/(2 h)$, and {\it forward} approximation $A^F_2: \ A^F_2 C(x) = [ -3 C(x) + 4 C(x+h) - C(x-2h)]/(2 h)$. Also $I$ denotes a unit matrix.
All these definitions assume that we work on a uniform grid, however this could be easily generalized for the non-uniform grid as well, see, e.g., \cite{HoutFoulon2010}.
Below we consider two cases, $\beta < 0$ and $\beta \ge 0$.
\begin{proposition} \label{propNIG+}
If $\beta < 0$, then the discrete counterpart $J$ of the operator $\mathcal{J}$ is the negative of an M-matrix if
\[ J = \delta \left( \sqrt{\alpha^2 - \beta^2}I - \left[(\alpha^2 - \beta^2)I - 2 \beta A^B_2 -
A^C_2 \right]^{1/2} \right) \]
The matrix $J$ is an $O(h^2)$ approximation  of the operator $\mathcal{J}$.
\end{proposition}
\begin{proof}
See Appendix.
\end{proof}
\begin{proposition} \label{propNIG-}
If $\beta \ge 0$, then the discrete counterpart $J$ of the operator $\mathcal{J}$ is the negative of an EM-matrix if
\[ J = \delta \left( \sqrt{\alpha^2 - \beta^2}I - \left[(\alpha^2 - \beta^2)I - 2 \beta A^F_2 -
A^C_2 \right]^{1/2} \right) \]
The matrix $J$ is an $O(h^2)$ approximation  of the operator $\mathcal{J}$.
\end{proposition}
\begin{proof}
See Appendix.
\end{proof}

Thus, according to Proposition~\ref{prop0} these finite difference schemes are unconditionally stable starting with some $N$, preserve positivity of the solution, and approximate operator $\mathcal{J}$ with $O(h^2)$. In our experiments shown below this positivity was achieved at $N  > 100$.

To complete the solution we need to compute the matrix exponential $ \mathcal{A} = e^{\Delta \tau \mathcal{J}}$ and the product  $z(x,\tau) = \mathcal{A} C(x)$. The last step, however, could be further simplified. To see this recall that the diffusion equations in \eqref{splitFin} have to be solved up to some order of approximation in time $\tau$. Suppose for this purpose we want to use a finite-difference scheme that provides a second order approximation, $O((\Delta \tau^2))$. However, \eqref{EqMerton} gives an {\it exact} solution of the corresponding pure jump equation (the second step in Strang's splitting scheme). Since Strang's scheme guarantees only second-order accuracy ($O((\Delta \tau)^2)$) to the exact solution of the full PIDE, the second step could be computed to the same order of accuracy.

To this end we can use the (1,1) P{\'a}de approximation of $e^{\Delta \tau \mathcal{J}}$,
\begin{equation} \label{mer1}
e^{\Delta \tau \mathcal{J}} \approx [1 - \frac{1}{2}\Delta \tau \mathcal{J}]^{-1}[1 + \frac{1}{2}\Delta \tau \mathcal{J}] + O(\Delta \tau^3).
\end{equation}
This could be re-written in the form of the fixed point Picard iterations scheme
\[ C^{(1)}(x, \tau + \Delta \tau) - C^{(1)}(x, \tau) = \frac{1}{2} \Delta \tau \mathcal{J}^* \left[C^{(1)}(x, \tau + \Delta \tau) + C^{(1)}(x, \tau)\right], \]
\noindent and this equation could be solved iteratively starting with the initial guess
$C^{(1)}(x, \tau + \Delta \tau) = C^{(1)}(x, \tau)$. Note that at each iteration the vector $z(x,t)$ should be computed.

\paragraph{Numerical experiments}
In our numerical experiments we consider the NIG model which also has a diffusion component uncorrelated with the jumps. We compute just one step of the splitting procedure, i.e. the jump integral, and don't consider solution of the diffusion part of the model.
We want to price a European call option and take the option model parameters similar to \cite{Halluin2005b}, i.e., $S_0 = K = 100, r = 0.05, \sigma = 0.15$. The NIG model parameters are $\delta = 0.2, \alpha = 10, \beta = - 5.7, \mu = 0$. One step in time is computed by taking $T = \Delta \tau = 0.01$. As $C^{(1)}(x, \Delta \tau)$ in the \eqref{splitFin} comes after the first step of splitting, we get it by using the Black-Scholes formula with the forward interest rate $r+c$ where the term $c = - \log \Phi_{NIG}(\alpha, \beta,\delta,\mu,-i)$ comes from the last term in the jump integral in \eqref{intGen}. At the second step the solution of the jump part $C^{(2)}_j(x,\Delta \tau)$ is produced given the initial condition $C^{(1)}(x,\Delta \tau)$ from the previous step. We compare our solution for the jump step with that obtained with $N = 2300$ which is assumed to be close to the exact value\footnote{This method is not very accurate. But as the exact solution is not known, it provides a plausible estimate of the convergence.}. The finite-difference grid was constructed as follows: the diffusion grid was taken from $x^D_{min} = 10^{-3}$ to $x^D_{max} = 30 \max(S,K)$. The jump grid is a superset of the diffusion grid, i.e. it coincides with the diffusion grid at the diffusion domain and then extends this domain up to $x^J_{max} = \log (10^5)$. Here to simplify the convergence analysis we use an uniform grid with step $h$. However, non-uniform grid can be easily constructed as well, and, moreover, that is exactly what this algorithm was constructed for.
The results of such a test are given in Table~\ref{Tab1}. Here $C$ is the price in dollars, $N$ is the number of grid nodes, $\beta_i$ is the order of convergence of the scheme. The "exact" price obtained at $N = 2300$ is $C_{num}(\Delta \tau)$ =  0.756574. It is seen that the convergence order $\beta_i
= \log_2 \frac{C(i) - C_{num}}{C(i+1)-C_{num}}, \ i=1,2...$ of the scheme is asymptotically close to $O(h^2)$.
\begin{table}[!ht]
\begin{center}
\begin{tabular}{|c|l|r|l|r|}
\hline
$C$ & $h$ & $N$ & $\beta_i$ \cr
\hline
4.0657 & 0.2763100 & 51  & - \cr
\hline
2.1239 & 0.1381550 & 101 & 1.275 \cr
\hline
1.1914 & 0.0690776 & 201 & 1.653 \cr
\hline
0.8171 & 0.0345388 & 401 & 2.845 \cr
\hline
0.7597 & 0.0172694 & 801 & 4.295 \cr
\hline
0.7565 & 0.0086347 & 1601 & 5.005 \cr
\hline
\end{tabular}
\caption{Convergence of the proposed scheme for the NIG model with $\beta = -5.7 < 0$, $T = \Delta \tau = 0.01$}
\label{Tab1}
\end{center}
\end{table}
As a sanity check we can compare this value with the reference value obtained by pricing this model (one step) using FFT, which is $C_{FFT}(\Delta \tau)$ = 0.757782. Note, that $C_{FFT}(\Delta \tau)$ should not be exactly equal to $C_{num}(\Delta \tau)$ because we use two steps instead of three as in the Strang algorithm \footnote{Don't miss this with the accuracy of the whole 3 steps Strang's algorithm which is $O(\Delta \tau^2)$. The test validates just the convergence in $h$, not in $\Delta \tau$.} which are equivalent to the splitting scheme of the first order in $\Delta \tau$, i.e. it has an error $O(\Delta \tau)$. However, still $C_{FFT}(\Delta \tau)$ is close to $C_{num}(\Delta \tau)$.
The second experiment uses the same set of parameters, but now $\beta = 5.7 > 0$. The results are given in Table~\ref{Tab2}. The "exact" price obtained at $N = 2300$ is $C_{num}(\Delta \tau)$ =  0.76864. Again the convergence order $\beta_i$  of the scheme is close to $O(h^2)$.
\begin{table}[!ht]
\begin{center}
\begin{tabular}{|c|l|r|l|r|}
\hline
$C$ & $h$ & $N$ & $\beta_i$ \cr
\hline
4.0127 & 0.2763100 & 51 & - \cr
\hline
2.1896 & 0.1381550 & 101 & 1.191 \cr
\hline
1.1627 & 0.0690776 & 201 & 1.850 \cr
\hline
0.8393 & 0.0345388 & 401 & 2.480 \cr
\hline
0.7710 & 0.0172694 & 801 & 4.892 \cr
\hline
0.7685 & 0.0086347 & 1601 & 4.557  \cr
\hline
\end{tabular}
\caption{Convergence of the proposed scheme for the NIG model with $\beta = 5.7 > 0, \ T = \Delta \tau = 0.01$}
\label{Tab2}
\end{center}
\end{table}
In this test $C_{FFT}(\Delta \tau)$ = 0.76773 which is also close to $C_{num}(\Delta \tau)$.
As a final note, performance wise the principal square root of matrix is better to compute using
the product form of the Denman-Beavers iteration, see \cite{Denman1976} for the description of the algorithm.

\subsection{Generalized hyperbolic models.}

Generalized hyperbolic process constitutes a broad subclass of L\'evy processes
which are generated by generalized hyperbolic (GH) distributions. They were introduced
in \cite{BN78}. See also \cite{EberleinKeller:95} for the detailed survey on how the hyperbolic distributions are used in finance.

The Lebesgue density of the GH distribution is a 5-parameter function
\begin{equation}
f(\lambda, \alpha, \beta, \delta, \mu) = a(\lambda, \alpha, \beta, \delta, \mu) \left[\delta^2 + (x-\mu)^2 \right]^{\frac{\lambda - 1/2}{2}} K_{\lambda - 1/2}
\left( \alpha \sqrt{\delta^2 + (x-\mu)^2} \right) e^{\beta(x-\mu)},
\end{equation}
\noindent where the normalization constant reads
\[ a(\lambda, \alpha, \beta, \delta, \mu) = \dfrac{(\alpha^2 - \beta^2)^{\lambda/2}}{\sqrt{2\pi} \alpha^{\lambda - 1/2} \delta^\lambda K_\lambda(\delta \sqrt{\alpha^2 - \beta^2})}.\]
Here $\alpha > 0$ determines the shape of the distribution, $\beta$ determines the skewness, and $0 \le |\beta| < \alpha$, $\mu \in \mathbb{R}$ determines the location (mean), $\delta > 0$ is scaling, and $\lambda \in \mathbb{R}$ determines the weight of the distribution in the tails. In particular, $\lambda = -1/2$ corresponds to the NIG distribution considered in the previous sections

The characteristic exponent of the GH process is
\begin{align} \label{GH_CE}
\phi(u,\alpha,\beta,\delta,\mu, \lambda) &= i u \mu  + \log \Psi \\
\Psi &= \left( \dfrac{\alpha^2 - \beta^2}{\alpha^2 - (\beta + i u)^2} \right)^{\lambda/2}
\dfrac{ K_\lambda(\delta \sqrt{\alpha^2 - (\beta + i u) ^2})}{ K_\lambda(\delta \sqrt{\alpha^2 - \beta^2})}. \nonumber
\end{align}
\noindent and the L\'evy density $\nu(dx)$ of the GH L\'evy motions reads
\begin{equation*}
\nu(dx) = \dfrac{\exp(\beta x)}{|x|} \left(
\int_0^\infty
\dfrac{ \exp(-\sqrt{2y + \alpha^2}|x|)}{\pi^2 y (J^2_{|\lambda|} (\delta \sqrt{2y}) + Y^2_{|\lambda|} (\delta \sqrt{2y}))}
d y + \mathds{1}_{\lambda \ge 0} \lambda e^{-\alpha |x|}
\right)dx,
\end{equation*}
\noindent where $J,Y$ are the corresponding Bessel functions.

From \eqref{MGT},\eqref{GH_CE} we immediately obtain that
\[ \mathcal{J} = \phi(-i \triangledown, \alpha,\beta,\delta,\mu, \lambda), \]
and the jump evolution equation (the second equation in \eqref{splitFin}) becomes
\begin{align} \label{ghEq}
C(x,\tau + \Delta \tau) &=  \mathcal{A} C(x,\tau), \\
\mathcal{A} & = e^{t \mu \triangledown} \left( \dfrac{\alpha^2 - \beta^2}{\alpha^2 - (\beta + \triangledown)^2} \right)^{\lambda \Delta \tau}
\left\{ \dfrac{ K_\lambda(\delta \sqrt{\alpha^2 - (\beta + \triangledown)^2})}{ K_\lambda(\delta \sqrt{\alpha^2 - \beta^2})} \right\}^{\Delta \tau}. \nonumber
\end{align}

Similar to the Meixner model (see the previous section) the first term $e^{t \mu \triangledown}$ could be taken out of this expression and moved to the diffusion part of \eqref{oper}.
The remaining operator $\mathcal{B} = e^{-t \mu \triangledown}\mathcal{A} $ could be represented in a form of a product of two operators
\begin{align} \label{ghFactor}
\mathcal{B} &= \mathcal{B}_1 \mathcal{B}_2, \\
\mathcal{B}_1 &= \left( \dfrac{\alpha^2 - \beta^2}{\alpha^2 - (\beta + \triangledown)^2} \right)^{\lambda \Delta \tau},
\qquad
\mathcal{B}_2 = \left\{ \dfrac{ K_\lambda(\delta \sqrt{\alpha^2 - (\beta + \triangledown)^2})}{ K_\lambda(\delta \sqrt{\alpha^2 - \beta^2})} \right\}^{\Delta \tau}. \nonumber
\end{align}
To construct the approximation of $\mathcal{B}_2$ take into account that by \cite{as64} for the modified Bessel functions of the third kind we have
\begin{align} \label{Kas}
K_\nu (z) &= \sqrt{\dfrac{\pi}{2 z}} e^{-z} \sum_{k=0}^\infty \dfrac{a_k(\nu)}{z^k}, \qquad |z| \rightarrow \infty, \ |\mbox{arg} \ z| < \dfrac{3 \pi}{2}, \\
a_0(\nu) &= 1, \ a_k(\nu) = \dfrac{(4 \nu^2 - 1^2)....(4 \nu^2 - (2k-1)^2)}{k! 8^k}. \nonumber
\end{align}
We want to approximate $K_\nu (z)$ up to $O(h^2)$. All operators $\mathcal{B}_i, i=1,2$ in \eqref{ghFactor} are actually the operator functions of another operator $z(\triangledown)$, where $z(x) \equiv \delta \sqrt{\alpha^2 - (\beta + x)^2}$. Obviously, any order discretization of the operator $\triangledown$ on a grid is proportional to $1/h$. Therefore, discretization of
$z$  is also proportional to $1/h$. Hence, such discretization applied to the terms $1/z^k$ will be proportional to $h^k$. When this operator affects a discrete vector function also defined at the same grid, the total error will be not be worth than $O(h^2)$. That means that in the series expansion \eqref{Kas} we can keep only terms with $k=0,1$ while omitting the remaining ones.

With allowance for that and \eqref{Kas} we redefine $\mathcal{B}_i, \ i=1,2$ as follows.
\begin{align} \label{redefB}
\mathcal{B} &= \mathcal{B}_1 \mathcal{B}_2, \qquad
\mathcal{B}_1 = \left( \dfrac{z(0)}{z(\triangledown)}\right)^{(\lambda + 1/2) \Delta \tau}, \\
\mathcal{B}_2 &= \left[e^{z(0) - z(\triangledown)}
\left( 1 + \dfrac{4 \lambda^2 - 1}{z(\triangledown)} \right) \left(\sum_{k=0}^\infty \dfrac{a_k(\nu)}{z(0)^k}\right)^{-1}\right]^{\Delta \tau}. \nonumber
\end{align}

Thus, now we need to construct an appropriate discretization of the operators $\mathcal{B}_i, \ i=1,2$. In doing that consider two cases.

\subsubsection{$\lambda \ge -1/2$}.

\begin{proposition} \label{Beta1}
Assume that $\beta < 0$. Denote by $B_1$ the following discrete representation of the operator $\mathcal{B}_1$ on a given grid $\bf{G}(x)$:
\[ B_1 = \left[ \dfrac{(\alpha^2 - \beta^2)I - 2 \beta A^B_2 - A^C_2}{\alpha^2 - \beta^2} \right]^{-\Delta \tau (\lambda +1/2)/2} \]
Then $B_1$ is a nonnegative matrix with all eigenvalues $|\lambda_i| < 1, \ \forall i \in [1,N]$. The matrix $B_1$ is an $O(h^2)$ approximation  of the operator $\mathcal{B}_1$.
\end{proposition}
\begin{proof}
See Appendix.
\end{proof}
Note that for a non-uniform grid the prove could be constructed in a similar way, but requires many technical details which we do not consider in this paper.

\begin{proposition} \label{Beta1P}
Assume that $\beta \ge 0$. Denote by $B_1$ the following discrete representation of the operator $\mathcal{B}_1$ on a given grid $\bf{G}(x)$:
\[ B_1 = \left[\dfrac{(\alpha^2 - \beta^2)I - 2 \beta A^F_2 - A^C_2}{\alpha^2 - \beta^2} \right]^{-\Delta \tau (\lambda +1/2)/2} \]
Then $B_1$ is a nonnegative matrix with all eigenvalues $|\lambda_i| < 1, \ \forall i \in [1,N]$. The matrix $B_1$ is an $O(h^2)$ approximation  of the operator $\mathcal{B}_1$.
\end{proposition}
\begin{proof}
The proof is analogous to the previous Proposition with allowance for the fact that matrix $-2 \beta A_2^B$ with $\beta < 0$ is the transpose of the matrix $-2 \beta A_2^F$ with $\beta \ge 0$. $\blacksquare$
\end{proof}

Now observe that the first operator $e^{z(0) - z(\triangledown)}$ in the definition of $\mathcal{B}_2^{\Delta \tau}$ is exactly the operator $\mathcal{A}$ in \eqref{EqMerton}. Therefore, Propositions~\ref{propNIG+} and \ref{propNIG-} could be used to construct the corresponding discretizations. The second part (a product of two terms in parentheses) could be represented as
\begin{equation} \label{twoPar}
C = \gamma \left( 1 + \dfrac{4 \lambda^2 - 1}{z(\triangledown)} \right) =
\gamma \left[ 1 + (4 \lambda^2 - 1) \left(\dfrac{1}{z^2(\triangledown)}\right)^{1/2} \right], \qquad \gamma = \left(\sum_{k=0}^\infty \dfrac{a_k(\nu)}{z(0)^k}\right)^{-1}.
\end{equation}
Operator $1/z^2(\triangledown)$ could be discretized using Propositions~\ref{Beta1} and \ref{Beta1P}. Coefficient $\gamma$ guarantees that all eigenvalues of the discrete discretization matrix are less than one. Thus, the proposed scheme is unconditionally stable.

Obviously as each operator $\mathcal{B}_i, \ i=1,2$ has the necessary property, a "product" of three operators (the consecutive application of them one after another) will result into the combined operator having same properties. This finalizes our construction.

Also notice, that at $\lambda = -1/2$, which corresponds to the NIG model considered in the previous sections, our scheme exactly translates to the scheme proposed there.

\subsubsection{$\lambda < -1/2$}.
At the first glance in this case the previous scheme will not converge as the eigenvalues of the  operator $\mathcal{B}_1$ are greater than one. Nevertheless, as shown below it still can be used. To demonstrate that for the sake of convenience re-write $\mathcal{B}$ in \eqref{redefB} in the form
\begin{align} \label{redefB2}
\mathcal{B} &= \left\{ e^\mathcal{M} \left( 1 + \dfrac{4 \lambda^2 - 1}{z(\triangledown)} \right) \left(\sum_{k=0}^\infty \dfrac{a_k(\nu)}{z(0)^k}\right)^{-1}\right\}^{\Delta \tau}, \nonumber \\
\mathcal{M} &=  z(0) - z(\triangledown) - (\lambda+1/2)\log \dfrac{z(\triangledown)}{z(0)}.
\end{align}
\noindent where now $-(\lambda + 1/2) > 0$. So the difference with the case $\lambda \ge -1/2$ is only in that we moved the operator $\mathcal{B}_1$ into the exponent term. Because of that for the terms in parentheses we leave the same discretization as in Propositions ~\ref{Beta1}, \ref{Beta1P}. Then the following Proposition is in order.
\begin{proposition} \label{NegLambdaBeta1}
Assume that $\beta < 0$. Denote by $Z$ the following discrete representation of the operator $z(\triangledown)$ on a given grid $\bf{G}(x)$:
\[ Z = \delta[(\alpha^2 - \beta^2)I - 2 \beta A^B_2 - A^C_2]^{1/2} \]
Then
\begin{align*}
B &= \left(\sum_{k=0}^\infty \dfrac{a_k(\nu)}{z(0)^k}\right)^{-\Delta \tau}
\left\{ e^M \left[ 1 + (4 \lambda^2 - 1)Z^{-1} \right] \right\}^{\Delta \tau}, \nonumber \\
M &=  z(0) - Z - (\lambda+1/2)\log \dfrac{Z}{z(0)}.
\end{align*}
\noindent is a nonnegative matrix with all eigenvalues $|\lambda_i| < 1, \ \forall i \in [1,N]$. The matrix $B$ is an $O(h^2)$ approximation  of the operator $\mathcal{B}$ in \eqref{redefB2}.
\end{proposition}
\begin{proof}
See Appendix.
\end{proof}

\begin{proposition} \label{NegLambdaBeta1P}
Assume that $\beta \ge 0$. Denote by $Z$ the following discrete representation of the operator $z(\triangledown)$ on a given grid $\bf{G}(x)$:
\[ Z = \delta[(\alpha^2 - \beta^2)I - 2 \beta A^F_2 - A^C_2]^{1/2} \]
Then
\begin{align*}
B &= \left(\sum_{k=0}^\infty \dfrac{a_k(\nu)}{z(0)^k}\right)^{-\Delta \tau}
\left\{ e^M \left[ I + (4 \lambda^2 - 1)Z^{-1} \right] \right\}^{\Delta \tau}, \nonumber \\
M &=  z(0) - Z - (\lambda+1/2)\log \dfrac{Z}{z(0)}.
\end{align*}
\noindent is a nonnegative matrix with all eigenvalues $|\lambda_i| < 1, \ \forall i \in [1,N]$. The matrix $B$ is an $O(h^2)$ approximation  of the operator $\mathcal{B}$ in \eqref{redefB2}.
\end{proposition}
\begin{proof}
The proof is analogous to the previous Proposition with allowance for the fact that matrix $-2 \beta A_2^B$ with $\beta < 0$ is the transpose of the matrix $-2 \beta A_2^F$ with $\beta \ge 0$. $\blacksquare$
\end{proof}

\paragraph{Numerical experiments}
The numerical experiments are provided similar to how we did it with the models considered in the previous sections. The GH model parameters are $\alpha = 10, \beta = -5.7, \delta = 0.2, \mu = 0$. The results for $\lambda = -1$ are given in Table~\ref{Tab4}. The "exact" price obtained at $N = 2100$ is $C_{num}(\Delta \tau)$ =  0.73580. The observed convergence order $\beta_i$ of the scheme is close to $O(h^2)$. Also $C_{FFT}(\Delta \tau)$ = 0.73746 which is also close to $C_{num}(\Delta \tau)$.
\begin{table}[!ht]
\begin{center}
\begin{tabular}{|c|l|r|l|r|}
\hline
$C$ & $h$ & $N$ & $\beta_i$ \cr
\hline
4.0492 & 0.2763100 & 51  & - \cr
\hline
2.1145 & 0.1381550 & 101 & 1.265 \cr
\hline
1.1701 & 0.0690776 & 201 & 1.666 \cr
\hline
0.7979 & 0.0345388 & 401 & 2.805 \cr
\hline
0.7389 & 0.0172694 & 801 & 4.325 \cr
\hline
0.7357 & 0.0086347 & 1601& 5.644\cr
\hline
\end{tabular}
\caption{Convergence of the proposed scheme for the GH model, $\lambda = -1$..}
\label{Tab4}
\end{center}
\end{table}

The second test which results are given in Table~\ref{Tab5} uses $\lambda = 1$. The "exact" price obtained at $N = 2100$ is $C_{num}(\Delta \tau)$ =  0.84440. The observed convergence order $\beta_i$ of the scheme is also close to $O(h^2)$. The FFT price $C_{FFT}(\Delta \tau)$ = 0.846985 is at the distance of 0.3\% from $C_{num}(\Delta \tau)$ while theoretically the error should be proportional to $O(\Delta \tau)$ = 1\%. Thus, it seems very reasonable that this error is due to the fact that we actually used the first order approximation in time in this test (see a detailed explanation for doing that in Section~\ref{NIGsect}.

\begin{table}[!ht]
\begin{center}
\begin{tabular}{|c|l|r|l|r|}
\hline
$C$ & $h$ & $N$ & $\beta_i$ \cr
\hline
4.1444 & 0.2763100 & 51  & - \cr
\hline
2.1657 & 0.1381550 & 101 & 1.320 \cr
\hline
1.2827 & 0.0690776 & 201 & 1.592 \cr
\hline
0.8978 & 0.0345388 & 401 & 3.038 \cr
\hline
0.8475 & 0.0172694 & 801 & 4.107 \cr
\hline
0.8443 & 0.0086347 & 1601& 5.490 \cr
\hline
\end{tabular}
\caption{Convergence of the proposed scheme for the GH model, $\lambda = 1$.}
\label{Tab5}
\end{center}
\end{table}

%%%%%%%%%%%%%%%%%%%%%%%%%%%%%%%%%%%%%%%%%%%%%%%%%
\subsection{Meixner model}

Meixner jump model was introduced by \cite{Schoutens98, Schoutens01}. It is built based on the Meixner distribution which belongs to the class of the infinitely divisible distributions. Therefore, it gives rise to a L\'evy process - the Meixner process. The Meixner process is flexible and analytically tractable, i.e. its pdf and CF are known in closed form.

The density of the Meixner distribution $f(a,b,d,m)$ reads
\begin{equation*}
f(x;a,b,d,m) = \dfrac{[2 \cos(b/2)]^{2 d} }{2 a \pi \Gamma(2d)} \exp\left[ \dfrac{b(x-m)}{a}\right] \left| \Gamma\left( d + \dfrac{i(x-m)}{a}\right) \right|^2,
\end{equation*}
\noindent where $a > 0, \ -\pi < b < \pi - a, \ d > 0, \ m \in \mathbb{R}$.

The characteristic exponent of the Meixner process is
\begin{equation} \label{MCE}
\phi(u, a,b,d,m) = 2d \left\{ \log [\cos(b/2)] - \log \left[ \cosh\left( \dfrac{a u - i b}{2} \right) \right] \right\} + i m u,
\end{equation}
\noindent and the L\'evy density $\nu(dx)$ of the Meixner process reads
\begin{equation*}
\nu(dx) = d \dfrac{\exp(bx/a)}{x \sinh(\pi x/a)}dx.
\end{equation*}

From \eqref{MGT},\eqref{MCE} we immediately obtain
\begin{equation} \label{MJ}
\mathcal{J} = \phi(-i \triangledown, a,b,d,m) =  2d \left\{ \log [\cos(b/2)] - \log \left[ \cos\left( \dfrac{a \triangledown + b}{2} \right) \right] \right\} + m \triangledown.
\end{equation}

Now observe that the last term $m \triangledown$ could be taken out and moved to the diffusion part of \eqref{oper}. This is because when constructing our splitting algorithm we have a freedom to decide which terms to keep under the jump part and which ones should be moved to the diffusion part. As the term $m \triangledown$ is proportional to $\triangledown$, i.e. it looks similar to the drift term of the diffusion part, we can naturally add it to the drift and eliminate it from the jump integral assuming that the remaining expression of $\mathcal{J}$ is well-defined.

Using the remaining part of \eqref{MCE} the operator $\mathcal{A} = e^{\Delta \tau \mathcal{J}}$ can be represented in the form
\begin{equation} \label{M_A}
\mathcal{A} = [\cos(b/2)]^{2 d \Delta \tau} \left[\sec\left( \dfrac{a \triangledown + b}{2} \right)\right]^{2 d \Delta \tau}
\end{equation}
Thus, our goal is to compute the product $\mathcal{A} C(x,\tau)$. 

To do that let us use a representation of $\cos(x)$ in a form of the infinite product, see \cite{as64}
\begin{equation*}
\cos(x) = \prod_{n=1}^\infty \left[1 - \dfrac{x^2}{\pi^2 (n- 1/2)^2}\right]
\end{equation*}
Then \eqref{M_A} could be re-written as follows
\begin{equation} \label{M_A1}
\mathcal{A} = [\cos(b/2)]^{\kappa} \prod_{n=1}^\infty \mathcal{A}_n, \qquad 
\mathcal{A}_n \equiv \left[1 - \dfrac{(a \triangledown + b)^2}{4 \pi^2 (n- 1/2)^2}\right]^{-\kappa}, \quad \kappa = 2 d \Delta \tau.
\end{equation}

The following Proposition now gives the solution of our problem.
\begin{proposition} \label{propMeixner}
Assume that $b < 0$. Denote by $M$ the following discrete representation of the operator 
\[ z(\triangledown) = 
1 - \dfrac{(a \triangledown + b)^2}{4 \pi^2 (n- 1/2)^2} \]
on a given grid $\bf{G}(x)$:
\[ M_n = I - \dfrac{1}{4 \pi^2 (n- 1/2)^2} \left(a^2 A^C_2 + 2 a b A^B_2 + b^2 I\right) \]
Then
\begin{align*}
B &=  [\cos(b/2)]^{\kappa} \prod_{n=1}^\infty M_n^{-\kappa}
\end{align*}
\noindent is a nonnegative matrix with all eigenvalues $|\lambda_i| < 1, \ \forall i \in [1,N]$. The matrix $B$ is an $O(h^2)$ approximation  of the operator $\mathcal{A}$ in \eqref{M_A}.
\end{proposition}
\begin{proof}
See Appendix.
\end{proof}
A similar Proposition is in place when $0 < b < \pi -a$.
\begin{proposition} \label{propMeixner1}
Assume that $0 < b < \pi -a$. Denote by $M$ the following discrete representation of the operator $z(\triangledown)$ on a given grid $\bf{G}(x)$:
\[ M_n = I - \dfrac{1}{4 \pi^2 (n- 1/2)^2}\left(a^2 A^C_2 + 2 a b A^F_2 + b^2 I\right) \]
Then
\begin{align*}
B &=  [\cos(b/2)]^{\kappa} \prod_{n=1}^\infty M_n^{-\kappa}
\end{align*}
\noindent is a nonnegative matrix with all eigenvalues $|\lambda_i| < 1, \ \forall i \in [1,N]$. The matrix $B$ is an $O(h^2)$ approximation  of the operator $\mathcal{A}$ in \eqref{M_A}.
\end{proposition}
\begin{proof}
See Appendix.
\end{proof}
Once all matrices $M_n$ are constructed, a sequential product 
\[ B C(x,\Delta \tau) = [\cos(b/2)]^{\kappa} \left( ...M_n^{-\kappa} \left( M_{n-1}^{-\kappa} ... \left(
 M_{2}^{-\kappa} \left(M_{1}^{-\kappa} C(x,\Delta \tau) \right) \right) \right) \right) \]
can be computed by consequently multiplying a matrix by vector product. FFT can be used for this purpose, see \cite{WangWanForsyth2007} for more details. 

For practical applications the infinite product in \eqref{M_A1} should be truncated. For instance $p = 5$ already provides a relative accuracy better than 1\%.

However, in many situations a more efficient method could be proposed. Assume that $0 \le 2 d \Delta \tau \le 2$. Indeed, this is always the case as $\Delta \tau \ll 1$ while the values of $d$ found, e.g., in \cite{SchoutensLevyBook2003} when calibrating the Meixner model to the option market data, were about $d = 50$. Obviously, decreasing  $\Delta \tau$ we can always make the above inequality to be valid, if necessary. However, this is not an attractive way to proceed, so below we are under believe that for reasonable values of $\Delta \tau$ this inequality is correct.

If so, then a variation of the method proposed in \cite{ItkinCarr2012Kinky} could be applied. The idea of the method is to consider the discrete operator $B$ as a function of the parameter $k$
\[ B(\kappa) =  [\cos(b/2)]^{\kappa} \prod_{n=1}^\infty M_n^{-\kappa}  \]
According to our assumption $0 \le \kappa \le 2$. Therefore, we can compute three vectors $z_0 = B(0)C(x,\tau)$, $z_1 = B(1)C(x,\tau)$, and $z_2 = B(2)C(x,\Delta \tau)$ and then interpolate them point-wise to $\kappa = 2 d \Delta \tau$. Monotonic spline interpolation could be used for this purpose. Again, see \cite{ItkinCarr2012Kinky} for more details, as well as the Theorem about a continuity of the price in $\kappa$ space proven there. 

It is easy to see that $z_0 = C(x, \tau)$. Now observe, that at $\kappa=1$ for every $n=1,2,...$ vector $z_{1,n}$ solves the following system of linear equations
\begin{equation} \label{M1}
M_n z_{1,n}(x,\tau + \Delta \tau) = z_{1,n-1}(x,\tau) 
\end{equation}
\noindent where $z_{1,1}(x,\tau) = C(x, \tau)$. Matrix $M_n$ by construction is a banded matrix with 4 non-zero diagonals. Also according to the Propositions~\ref{propMeixner}, \ref{propMeixner1} it is an EM-matrix, therefore \eqref{M1} is well-defined. As the complexity of this solution is $O(N)$, the total complexity of this step at $\kappa = 1$ is $ O(p N)$. This is worse that a pure linear convergence, but still could be better than that of the FFT. After the $p$ steps are done the final step is to multiply $z_{1,p}$ by $[\cos(b/2)]$.

Similarly, at $\kappa=2$ vector $z_{2,n}$ solves
\begin{equation} \label{M2}
M^2_n z_{2,n}(x,\tau + \Delta \tau) = z_{2,n-1}(x,\tau) 
\end{equation}
\noindent where $z_{2,1}(x,\tau) = C(x, \tau)$. 
Matrix $M^2_n$ by construction is a banded matrix with 7 non-zero diagonals, and also is an EM-matrix. Therefore, \eqref{M2} is well-defined and could be solved with the complexity $O(N)$. The total complexity of this step is also $ O(p N)$, hence the total complexity of the entire algorithm is $O(2 p N)$. 

Finally point-wise interpolation of three vectors to the given value of $\kappa$ has the complexity $O(N)$ if the interpolation coefficients are precomputed. Therefore, the total complexity of the method is still $2p O(N)$.\footnote{As in our case the banded matrices have either 4 or 7 diagonals, this complexity is also proportional to some coefficient which in this case could be about 4-7 per one solution for a given $n$. Still as was checked in our tests the method is faster than the FFT at relatively small $p = 3-5$, see below.}

\paragraph{Numerical experiments}
In our numerical experiment the values of the Meixner model parameters are taken as suggested in \cite{SchoutensLevyBook2003}, i.e.
$a = 0.04; b = -0.32754, d=52$, but here we use $m = 0$. Other parameters are the same as in the previous sections.

The results obtained with the first method at $p = 10$ are given in Table~\ref{TabM1}. The "exact" price obtained at $N = 3200$ is $C_{num}(\Delta \tau)$ =  1.0068. The observed convergence order $\beta_i$ of the scheme is close to $O(h^2)$ up to the point $N=1601$ where it drops down. To check what is the problem we ran the second test because further increase of $N$ when computing the matrix exponential gives rise to the RAM capacity of our PC being exceeded, while for the second method this is not a problem due to the banded structure of all matrices.

\begin{table}[!ht]
\begin{center}
\begin{tabular}{|c|l|r|l|r|}
\hline
$C$ & $h$ & $N$ & $\beta_i$ \cr
\hline
4.1791 & 0.2763100 & 51 &  - \cr
\hline
2.1771 & 0.1381550 & 101 & 1.439 \cr
\hline
1.3661 & 0.0690776 & 201 & 1.7045 \cr
\hline
1.0328 & 0.0345388 & 401 & 3.789 \cr
\hline
1.0060 & 0.0172694 & 801 & 5.022 \cr
\hline
1.0058 & 0.0086347 & 1601 & 0.322 \cr
\hline
\end{tabular}
\caption{Convergence of the proposed scheme for the Meixner model.}
\label{TabM1}
\end{center}
\end{table}
In this test $C_{FFT}(\Delta \tau)$ = 1.0145 which is also close to $C_{num}(\Delta \tau)$.\footnote{The standard FFT method at these values of parameters is very sensitive to the choice of the dumping factor $\alpha$, therefore this price was computed using the cosine method of \cite{FangOosterlee2008}.}

In the second test we repeated the previous one but now using our second approach - interpolation in the $\kappa$ domain. The results are given in Table~\ref{TabM2}.
\begin{table}[!ht]
\begin{center}
\begin{tabular}{|c|l|r|l|r|}
\hline
$C$ & $h$ & $N$ & $\beta_i$ \cr
\hline
4.1779 & 0.2763100 & 51 &  -  \cr
\hline
2.1765 & 0.1381550 & 101 & 1.439 \cr
\hline
1.3657 & 0.0690776 & 201 & 1.705 \cr
\hline
1.0327 & 0.0345388 & 401 & 3.812 \cr
\hline
1.0060 & 0.0172694 & 801 &  4.410 \cr
\hline
1.0058 & 0.0086347 & 1601 & 0.154 \cr
\hline
1.0068 & 0.0043174 & 3201 & 1.842 \cr
\hline
\end{tabular}
\caption{Convergence of the proposed scheme for the Meixner model - the "interpolation in $\kappa$" method.}
\label{TabM2}
\end{center}
\end{table}
 The "exact" price obtained at $N = 6401$ is $C_{num}(\Delta \tau)$ =  1.0072. Despite the convergence is close to $O(h^2)$ almost at all $h$, there is a spike at $N = 1601$ which indicates that monotonicity of the price as a function of $h$ changes close to this point.
 
 The typical time to compute the price in a single point $x$ using the cosine method with 12 terms in the expansion in our experiments was 4 msec. The time necessary for the interpolation method to compute $C(x,\tau + \Delta \tau)$ with $x$ now being a vector containing $N = 801$ points and $p=10$ was 4.2 msec. Observe, that if in this test we change $p$ to $p=5$ the total time accordingly halves the previous one, while the "exact" price obtained at $N = 6401$ becomes  $C_{num}(\Delta \tau)$ =  1.0031. In other words the difference is just 0.4\%.
 
We also regressed the computational time at $p=5$ to the number of grid points $N$ to check the order of convergence of the method. The results are given below in Table~\ref{TabM3}.
 \begin{table}[!ht]
 \begin{center}
 \begin{tabular}{|c|c|c|c|c|c|}
 \hline
 $N$ & 101 & 201 & 401 & 801 & 1601 \cr
 \hline
$\beta$ & 0.529 & 0.382 & 1.031 & 2.017 & 0.599 \cr
 \hline
 \end{tabular}
 \caption{Regression of the elapsed time $t_i$ for the interpolation method to the number of grid points $N_i$.}
 \label{TabM3}
 \end{center}
 \end{table}
 Here $\beta = \log_2 \left( \dfrac{t_{i} - t_{i+1}}{t_{i+1} - t_{i+1}} \right)$, $T_i$ is the elapsed time when $N = N_i$. It is known that this method usually is not very accurate in the estimate of the complexity, however we don't have a better one. It is seen that except the point at $N=801$ the complexity is about $O(N)$.

%%%%%%%%%%%%%%%%%%%%%%%%%%%%%%%%%%%%%%%%%%%%%%%%%%%%%%%%%%%%
\section{Conclusion}
%%%%%%%%%%%%%%%%%%%%%%%%%%%%%%%%%%%%%%%%%%%%%%%%%%%%%%%%%%%%
This paper is a further extension of \cite{Itkin2014}) where a new method of solving jump-diffusion PIDEs was proposed. This method exploits a number of ideas, namely:

\begin{enumerate}
\item First, we transform a linear non-local integro-differential operator (jump operator) into a local nonlinear (fractional) differential operator. Thus, the whole jump-diffusion operator $\mathcal{J} + \mathcal{D}$ is represented as a sum of the linear and non-linear parts.
\item Second, operator splitting on financial processes\footnote{This is similar to splitting on physical processes, e.g., convection and diffusion, which is well-known in computational physics.} is applied to this operator, namely splitting a space operator into diffusion and jumps parts. For nonlinear operators, this approach was elaborated on based on the definition of Lie derivative (see \cite{ThalhammerKoch2010}). The described splitting scheme provides a second-order approximation of $\mathcal{J} + \mathcal{D}$ in time.
\item At the third step various finite-difference approximations of the non-linear differential operator $\mathcal{J}$ are proposed. In \cite{Itkin2014} Merton, Kou and GTSP (aka CGMY or KoBoL) models were considered. In this paper we demonstrated how to construct these approximations for the NIG, Hyperbolic and Meixner models to (i) be unconditionally stable, (ii) be of first- and second-order approximation in the space grid step size $h$ and (iii) preserve positivity of the solution. The results are presented as propositions, and the corresponding proofs are given based on modern matrix analysis, including a theory of M-matrices, Metzler matrices and eventually exponentially nonnegative matrices.
\end{enumerate}

All these results seem to be new. The method is naturally applicable to both uniform and nonuniform grids, and is easy for programming, since the algorithm is similar to all jump models. Also notice that the present approach allows pricing some exotic, e.g., barrier options as well. In addition, it respects not just vanilla but also digital payoffs. In principle, American and Bermudan options could also be priced by this method, however this requires some more delicate consideration which will be presented elsewhere.

\clearpage
%%%%%%%%%%%%%%%%%%%%%%%%%%%%%%%%%%%%%%%%%%%%%%%%%%%%%%%%%%%%
\section*{Acknowledgments}
%%%%%%%%%%%%%%%%%%%%%%%%%%%%%%%%%%%%%%%%%%%%%%%%%%%%%%%%%%%%
I thank Peter Carr and Peter Forsyth for very fruitful discussions, and Alex Lipton, Igor Halperin and Alexey Polishchuk for useful comments. I assume full responsibility for any remaining errors.

%==================================================
%\def\myBib{C:/AndreyItkin/Documents/Papers/aitkin_fin}
%\def\myBib{C:/AndreyItkin/MySettings2011/aitkin_fin}
%\def\myBib{aitkin_fin}
%\bibliographystyle{apalike}
%\bibliography{\myBib}
\newcommand{\noopsort}[1]{} \newcommand{\printfirst}[2]{#1}
  \newcommand{\singleletter}[1]{#1} \newcommand{\switchargs}[2]{#2#1}

%%%%%%%%%%%%%%%%%%%%%%%%%%%%%%%%%%%%%%%%%%%%%%%%%%%%%%%%%%%%
\clearpage
\appendix
\appendixpage

\section{Proof of Proposition~\protect{\ref{propNIG+}}}
%%%%%%%%%%%%%%%%%%%%%%%%%%%%%%%%%%%%%%%%%%%%%%%%%%%%%%%%%%%%
To prove this proposition we need technique that we used in \cite{Itkin2014}. It is closely related to the concept of an ``eventually positive matrix'', see \cite{Noutsos2008}. Below we reproduce some definitions from this paper necessary for our further analysis.

\begin{definition}
An $N\times N$ matrix $A = [a_{ij}]$ is called
\begin{itemize}
\item {\it eventually nonnegative}, denoted by $A \overset{v}{\ge} 0$, if there exists a
positive integer $k_0$ such that $A^k \ge 0$ for all $k > k_0$; we denote the
smallest such positive integer by $k_0 = k_0(A)$ and refer to $k_0(A)$ as the power index
of $A$;

\item {\it exponentially nonnegative} if for all $t > 0, \ e^{t A} = \sum_{k=0}^\infty \frac{t^k A^k}{k!} \ge 0$;

\item {\it eventually exponentially nonnegative} if there exists $t_0 \in [0,\infty)$ such
that $e^{t A} \ge 0$ for all $t > t_0$. We denote the smallest such nonnegative
number by $t_0 = t_0(A)$ and refer to it $t_0(A)$ s the exponential index of $A$.
\end{itemize}
\end{definition}

We also need the following Lemma from \cite{Noutsos2008}:
\begin{lemma} \label{lemma1}
Let $A \in \mathbb{R}^{N\times N}$. The following are equivalent:
\begin{enumerate}
\item $A$ is eventually exponentially nonnegative.
\item $A + b I$ is eventually nonnegative for some $b \ge 0$.
\item $A^T + b I$ is eventually nonnegative for some $b \ge 0$.
\end{enumerate}
\end{lemma}

We also introduce a definition of an EM-matrix, see \cite{ElhashashSzyld2008}.
\begin{definition} \label{def1}
An $N\times N$ matrix $A = [a_{ij}]$ is called an EM-Matrix if it can be represented as $A = sI - B$ with $0 < \rho(B) < s$, $s > 0$ is some constant, $\rho(B)$ is the spectral radius of $B$, and $B$ is an eventually nonnegative matrix.
\end{definition}

For the following we need two Lemmas.
\begin{lemma} \label{lemma2}
Let $A \in \mathbb{R}^{N\times N}$, and $A = (\alpha^2 - \beta^2) I - 2 \beta A_2^B$. Then $A$ is an EM-matrix.
\end{lemma}
\begin{proof}
Denote $d_i$ the $i$-th upper diagonal of $A$. So $d_0$ means the main diagonal, etc.

1. First, show that $2 \beta A^B_2$ is an eventually exponentially nonnegative matrix. To see this use representation $e^{t \beta A_2^B} \equiv [e^{t B}]^{1/(2h)}$ where $B$ is a lower tridiagonal matrix with all $d_0$ elements equal to $3 \beta < 0$ , all $d_{-1}$ elements equal to $-4 \beta > 0$, and all $d_{-2}$ elements equal to $\beta < 0$.  Positivity of $e^{t B}$ can be verified explicitly at $t > N$. The intuition behind that is that the elements on $d_2$ are small in absolute values as compared with that of $d_1$. Taking the square of $B$ propagates large positive values on $d_1$ to the diagonal $d_2$. Taking the square of $B^2$ propagates them to $d_3$, etc.

From $h > 0$ it follows that $e^{2 t \beta A_2^B} \ge 0$, i.e. $2 \beta A_2^B$ is eventually exponentially nonnegative.

According to Lemma~\ref{lemma1}, the eventual exponential non-negativity of $\beta A^B_2$ means that there exists $b \ge 0$ such that $\beta A_2^B + b I = \frac{1}{2h}(B + 2 h b I)$ is eventually nonnegative for some $b \ge 0$. Let us denote $B_1 = B + 2 h b I$ and chose $b = 3/(2h) + \epsilon$, where $\epsilon \ll 1$. In practical examples we can choose $\epsilon = 10^{-6}$. Then $d_0(B_1) = \epsilon, d_1(B_1) = 2, d_2(B_1) = -1$. It is easy to check that $B_1^{N+3} \ge 0$. Again that is because $d_{-1}(B_1) > 0, |d_{-1}(B_1)| > |d_{-2}(B_1)|$, so taking the square of $B_1$ propagates large positive values on $d_{-1}$ to the diagonal $d_{-2}$, etc. Thus, $\beta A^B_2 + b I$ with $b = 3/(2h) + \epsilon$ is the eventually nonnegative matrix.

2. Represent $A$ as $A = (\alpha^2 - \beta^2 + b)I - (2 \beta A_2^B + b I)$. Observe, that $\rho(2 \beta A_2^B + b I) = \epsilon$ and $s = (\alpha^2 - \beta^2 + b) > \epsilon$ as $|\beta| < \alpha$. Thus, by definition, $A$ is an EM-matrix. $\blacksquare$
\end{proof}

For the later we will also need this Lemma:
\begin{lemma} \label{lemma21}
The inverse of the matrix $A \equiv sI - P$ is a nonnegative matrix.
\end{lemma}
\begin{proof}
Observe that all eigenvalues of $P$ are $\lambda_i = \epsilon, \ \forall i \in [1,N]$. Therefore $\rho(P) = \epsilon$. Following \cite{LeMcDonald2006} denote $index_\lambda(A)$ to be the degree of $\lambda$ as a root of the minimal polynomial of $A$. As matrix $P$ doesn't have zero eigenvalues in its spectrum $index_0 (P) = 0 < 1$.

Non-negativity of $A^{-1}$ then follows from the Theorem
\begin{theorem}[Theorem 4.2 in \cite{LeMcDonald2006}]
Let $P$ be an $N \times N$ irreducible eventually nonnegative matrix with $index_0(P) \le 1$,
then there exists $\mu > \rho(P)$ such that if $\mu > s > \rho(P)$, then $(s I - P)^{-1} \ge 0$.
\end{theorem}
To apply this Theorem choose any positive $\mu > s$.
\end{proof}

Now we are ready to prove the Proposition~\ref{propNIG+}.
\begin{proof}[Proof of Proposition~\ref{propNIG+}]
Recall, that in the Proposition~\ref{propNIG+} the following scheme is proposed
\begin{align*}
 J = \delta \left( \sqrt{\alpha^2 - \beta^2}I - \left[(\alpha^2 - \beta^2)I - 2 \beta A^B_2 -
A^C_2 \right]^{1/2} \right)
\end{align*}

We prove separately each statement of the proposition, namely:
\begin{enumerate}
\item The above scheme is $O(h^2)$ approximation of the operator $ L_R$;
\item Matrix $J$ is the negative of an EM-matrix.
\end{enumerate}

{\it Proof of (1):} This follows from the fact that $A^C_2$ is a central difference approximation of the operator $\triangledown^2$ to second order in $h$, while $A^B_2$ is the one-sided second order approximation.\\

{\it Proof of (2):} Matrix $M_1 = (\alpha^2 - \beta^2)I - 2 \beta A^B_2$ is an EM-matrix. Matrix $-A^C_2$ is an M-matrix. The sum of an EM-matrix and M-matrix is an EM-matrix. Therefore $M_2 = (\alpha^2 - \beta^2)I - 2 \beta A^B_2 - A^C_2$ is an EM-matrix. According to properties of M-matrices $M_2^{1/2}$ is also an EM-matrix. Then, $-M_2$ is the negative of an EM-matrix. Now adding a diagonal matrix $M_3 = \sqrt{\alpha^2 - \beta^2}I$ to $-M_2$ we still obtain the resulting matrix to be the negative of an EM-matrix. That is because $\beta < 0$, and thus diagonal elements of $d_0(M_3) < d_0(M_2)$. In other words, diagonal elements of $M_3 - M_2$ are negative. Finally, as $\delta > 0$ the entire matrix $J$ is the negative of an EM-matrix.
That means that starting with some $N$ matrix $e^{\Delta \tau J}$ is positive, and all eigenvalues of $J$ are negative. Therefore, the spectral norm of the operator $\mathcal{A} = e^{\Delta \tau J}$ (which is $\lambda = \max_i |\lambda_i|, \forall i \in [1,N]$, $\lambda_i$ are the eigenvalues  of $\mathcal{A}$) obeys $\lambda < 1$. That means that the proposed scheme is unconditionally stable starting from some $N$. That finalizes the proof. $\blacksquare$
\end{proof}

%%%%%%%%%%%%%%%%%%%%%%%%%%%%%%%%%%%%%%%%%%%%%%%%%%%%%%%%%%%%
\section{Proof of Proposition~\protect{\ref{propNIG-}}}
%%%%%%%%%%%%%%%%%%%%%%%%%%%%%%%%%%%%%%%%%%%%%%%%%%%%%%%%%%%%
This proposition could be proven in a similar way as we did it with Proposition~\ref{propNIG+}. The main observation here is that. Suppose we proved Proposition~\ref{propNIG+} with $\beta = \beta_1 < 0$. Now choose $\beta_2 = - \beta_1 > 0$. Then matrix $M_2 = (\alpha^2 - \beta_2^2)I - 2 \beta_2 A^F_2$ is the transpose of the matrix $M_1 = (\alpha^2 - \beta_1^2)I - 2 \beta_1 A^B_2$ in the previous proof. So this is an upper triangular EM matrix. The remaining proof follows the exactly same steps as in the previous Appendix.  $\blacksquare$.

%%%%%%%%%%%%%%%%%%%%%%%%%%%%%%%%%%%%%%%%%%%%%%%%%%%%%%%%%%%%
\section{Proof of Proposition~\protect{\ref{Beta1}}}
%%%%%%%%%%%%%%%%%%%%%%%%%%%%%%%%%%%%%%%%%%%%%%%%%%%%%%%%%%%%
By Lemma~\ref{lemma2} matrix $(\alpha^2 - \beta^2)I - 2 \beta A^B_2$ is an EM-matrix.
So is $-A^C_2$. A sum of two EM-matrices is an EM-matrix. Therefore, $M = \left[(\alpha^2 - \beta^2)I - 2 \beta A^B_2 - A^C_2 \right]^{-1}$ is an EM-matrix.
Non-negativity of $M^{-1}$ then follows from Lemma~\ref{lemma21}. Therefore, $B_1$ is a nonnegative matrix.

Second order approximation follows from the fact that $A^C_2$ is the second order central difference approximation of the second derivative, and $A_2^B$ is the second order one-sided approximation of the first derivative.

The last point to prove is that all eigenvalues $\lambda_i$ of an $\mathcal{B}_1$ have positive real parts and obey the condition $|\lambda_i| < 1, \ \forall i \in [1,N]$. First, argue some intuition behind this. Consider matrices $M_1 = - 2 \beta A^B_2, \ M_2 = (\alpha^2 - \beta^2)I - A^C_2$. On a uniform grid they both are Toeplitz matrices. It is known that asymptotically at $N \rightarrow \infty$ the Toeplitz matrices commute, see \cite{Gray2006}. For commuting matrices the eigenvalues of the sum are the sum of the eigenvalues. Now, $M_1$ is a lower triangular matrix with the eigenvalues being the values at the main diagonal, i.e. $\lambda_i = 3 |\beta| /h > 0, \ i \in [1,N]$. The eigenvalues of $M_2$ could be represented as $ \lambda_i = (\alpha^2 - \beta^2) + 2/h^2 + \Lambda_i, \ i \in [1,N]$, where $\Lambda_i$ are the eigenvalues of the matrix constructed of the first lower and upper diagonals of $M_2$ while all the other elements vanish. It is known that the eigenvalues of such a matrix are $\Lambda_i = - \dfrac{2}{h^2} \cos \dfrac{i \pi}{N+1}$. Therefore, the eigenvalues of $M$ are
\begin{equation} \label{eigM}
\lambda_i = (\alpha^2 - \beta^2) + 3 \dfrac{|\beta|}{h} + \dfrac{4}{h^2} \sin^2 \dfrac{i \pi}{2(N+1)} > \alpha^2 - \beta^2 > 0, \ i \in [1,N].
\end{equation}
Thus, they are positive. Also, based on this inequality the eigenvalues of $B_1$ are
\[ \left(\dfrac{\alpha^2 - \beta^2}{\lambda_i}\right)^ {\Delta \tau (\lambda +1/2)/2} < 1. \]
Thus, the latest statement of the Proposition is asymptotically correct at large $N$ for an uniform grid. For smaller $N$ this could not be the case. Note however, that in our numerical experiments $N=100$ was sufficient for $B_1$ to acquire this property. $\blacksquare$

%%%%%%%%%%%%%%%%%%%%%%%%%%%%%%%%%%%%%%%%%%%%%%%%%%%%%%%%%%%%
\section{Proof of Proposition~\protect{\ref{NegLambdaBeta1}}}
%%%%%%%%%%%%%%%%%%%%%%%%%%%%%%%%%%%%%%%%%%%%%%%%%%%%%%%%%%%%
First, we need the following Lemma:
\begin{lemma} \label{lemmaLog}
Let $A \in \mathbb{R}^{N\times N}$ be an M-matrix. By definition an M-Matrix can be represented as $A = sI - B$ with $0 < \rho(B) < s$, $s > 0$ is some constant, $\rho(B)$ is the spectral radius of $B$, and $B$ is a nonnegative matrix. Then $\log A$ is an M-matrix if $s - \rho(B) > 1$.
\end{lemma}
\begin{proof}
Represent $A$ in the form
\[
\log A = \log s + \log (I - B/s) = \log s - \sum_{k=1}^\infty \dfrac{1}{i s^k} B^k
\]
As $B$ is a nonnegative matrix, the sum in the right hand side of the above equality is also a nonnegative matrix. Hence, all elements of $\log A$ are nonpositive, except might be those on the main diagonal. Actually, all $d_0(\log A)$ elements are positive.

To see that observe that
\begin{equation*}
\rho \left(\sum_{k=1}^\infty \dfrac{1}{i s^k} B^k\right) = \sum_{k=1}^\infty \dfrac{1}{i s^k} \rho(B)^k = -\log \left[1 - \dfrac{\rho(B)}{s}\right]
\end{equation*}
\noindent and, therefore,
\[ \rho(\log A) = \log \left[s - \rho(B) \right] > 0, \]
\noindent if $s - \rho(B) > 1$.
$\blacksquare$
\end{proof}

\begin{corollary} \label{cor}
Let $A \in \mathbb{R}^{N\times N}$ be an EM-matrix. Then $\log A$ is an EM-matrix if $s - \rho(B) > 1$.
\end{corollary}
\begin{proof}
The proof directly follows the steps in the Proof of the above Lemma with allowance for the definition and properties of an EM-matrix. $\blacksquare$
\end{proof}

Now we can prove the Proposition~\ref{NegLambdaBeta1}. As same discretization of $z(\triangledown)$ as in Propositions~\ref{propNIG+} is used all we need is to prove two statements:
\begin{enumerate}
\item The discretization $D$ of $\log \dfrac{z(\triangledown)}{z(0)}$ should be an EM-matrix, if $D[z(\triangledown)]$ is an EM-matrix.

\item The real parts of the eigenvalues $\lambda_i(D[M]), \ i \in [1,N]$ should be negative, where $D[M]$ is the discretization of $M$.
\end{enumerate}

\paragraph{Proof of 1.} Consider the eigenvalues of $Z$ which could be found using \eqref{eigM}
\begin{equation*}
\lambda_i(Z) = \delta \left[(\alpha^2 - \beta^2) + 3 \dfrac{|\beta|}{h} + \dfrac{4}{h^2} \sin^2 \dfrac{i \pi}{2(N+1)}\right]^{1/2}
, \quad i \in [1,N].
\end{equation*}
From here
\[ \rho(Z) = \max_i \lambda_i > \delta \dfrac{\pi}{h}  > 1 \]
\noindent if $h < \delta \pi$.

Matrix $Z$ by construction is an EM-matrix, see Proof to Proposition~\ref{Beta1}. Therefore, it can be represented in the form $Z = sI - B$. Then
\[ 1 < \rho(Z) = s - \rho(B). \]
Now we are under assumptions of the Corollary~\ref{cor}, therefore  $\log{Z}$ is an EM-matrix.

As $\log z(0)I$ is a non-negative diagonal matrix, matrix $M_z = \log \left[ Z(z(0)I)^{-1} \right]$ is an EM-matrix. That is because $Z$ is an EM-matrix, and multiplication of $Z$ by $(z(0)I)^{-1}$ (which is a diagonal matrix with the diagonal elements $1/(\alpha^2 - \beta^2)$)
changes only the diagonal elements of $Z$. In other words, the diagonal elements of $M_z$ are $\lambda_i(Z)/(\alpha^2 - \beta^2) > 1$. $\blacksquare$.

\paragraph{Proof of 2.} The second property follows from the fact that the eigenvalues of $D[z(\triangledown)]$ are positive, see Proposition~\ref{Beta1}. If so, the principal matrix logarithm exists and is well-defined. Based on asymptotic properties of the Toeplitz matrices the eigenvalues of $D[M]$ are asymptotically equal to the sum of the eigenvalues of $D[z(0) - z(\triangledown)]$  and $D[- (\lambda+1/2)\log (z(\triangledown)/z(0))]$, see Proof to Proposition~\ref{Beta1} in Appendix. Also, based on \eqref{eigM} the eigenvalues $\Lambda_i$ of $D[z(\triangledown)]$ in the leading term are proportional to $1/h$. Therefore, in the leading term
\[ \lambda_i(D[M]) = - \delta \left(\dfrac{{\bar \Lambda}_i }{h} + (\lambda+1/2)\log \dfrac{{\bar \Lambda}_i}{h (\alpha^2 - \beta^2)} \right) \]
\noindent where ${\bar \Lambda}_i > 0$ is a part of $\Lambda_i$ which in the leading term doesn't depend on $h$. Now observe that the inequality
\[ \dfrac{1}{h} {\bar \Lambda}_i + (\lambda+1/2) \log \dfrac{ {\bar \Lambda_i}}{h(\alpha^2 - \beta^2)}  > 0 \]
\noindent with $\kappa = -(\lambda+1/2) > 0$ could be transformed to
\[ \dfrac{1}{h} {\bar \Lambda}_i - \kappa \log \dfrac{ {\bar \Lambda_i}}{h}  > 0 \]
if $b = \log (\alpha^2 - \beta^2) > 0$. It is always valid if $\kappa < e$, valid at
$h < {\bar \Lambda}_i/e $ if $\kappa = e$, and at $h < {\bar \Lambda}_i e^{W(-1/\kappa)}$ at $\kappa > e$, where $W_k(y)$ is the Lambert W-function, see \cite{HMF2010}.

If $b < 0$ we need to consider the  inequality
\[ \dfrac{1}{h} {\bar \Lambda}_i - \kappa \log \dfrac{ {\bar \Lambda_i}}{h}  > - \kappa b > 0 \]
\noindent which solution reads
\[ h < - \dfrac{{\bar \Lambda}_i}{\kappa W\left( - \dfrac{e^{-b^2/\kappa}}{\kappa } \right) }
\]
$\blacksquare$

Thus in both cases $b > 0$ and $b < 0$ there exists an upper boundary on the space step $h$ which, however, doesn't depend on step in time $\Delta \tau$. Therefore, in this sense the proposed scheme  is unconditionally stable in $h$ starting from some $h$ given in the solutions of the above inequalities. Numerical calculations show that this upper limit is not very restrictive unless we consider an extreme case when $\alpha \approx |\beta|$.

Combining all the above we conclude that $\lambda_i(D[M]) < 0, \ i \in [1,N]$. Therefore, the eigenvalues of the operator $e^{M}$ are nonnegative and $\lambda_i(D[e^{M}]) < 1, \ i \in [1,N]$.

The last point to prove is that matrix $M_2 = I + (4 \lambda^2 - 1)Z^{-1}$ is a nonnegative matrix with all positive eigenvalues less than one in the absolute value. This follows from the fact that: i) Z is an EM-matrix, ii) an inverse of the EM-matrix is an eventually nonnegative matrix, iii) all eigenvalues of $Z$ are positive, therefore so are the eigenvalues of $Z^{-1}$; iv) the eigenvalues of $Z$ are less than one, $\lambda(Z)_i > 1, \ \forall i \in [1,N]$, therefore $\lambda(Z^{-1})_i < 1, \ \forall i \in [1,N]$. All this properties were already proven in Proposition~\ref{Beta1}.

The entire statement of the Proposition now follows because the product of two nonnegative matrices is a nonnegative matrix. Also as eigenvalues of both matrices $M$ and $M_2$ are positive and less than one, consecutive application of them produces a convergent transformation with the same properties of the eigenvalues of the operator product. This finalizes the prove.

%%%%%%%%%%%%%%%%%%%%%%%%%%%%%%%%%%%%%%%%%%%%%%%%%%%%%%%%%%%%
\section{Proof of Proposition~\protect{\ref{propMeixner}}}
%%%%%%%%%%%%%%%%%%%%%%%%%%%%%%%%%%%%%%%%%%%%%%%%%%%%%%%%%%%%
First observe that the operators $M_n, \ n=1,2,...$ has the same structure as the operator $M$ in the beginning of the proof of Proposition~\ref{Beta1}. Therefore, $M_n$ is an EM-matrix, and $M_n^{-\kappa}$ with $\kappa > 0$ is the nonnegative matrix.

Second order approximation follows from the fact that $A^C_2$ is the second order central difference approximation of the second derivative, and $A_2^B$ is the second order one-sided approximation of the first derivative.

The next point to prove is that all eigenvalues $\lambda_{n,i}$ of an $M_n^{-\kappa}$ have positive real parts and obey the condition $|\lambda_{n,i}| < 1, \ \forall i \in [1,N]$. This also directly follows from the proof of Proposition~\ref{Beta1}.

Thus, at every $n$ the map $M_n^{-\kappa}: z_{n-1}(x,\tau) \rightarrow z_{n}(x,\tau) = M_n^{-\kappa} z_{n-1}(x,\tau)$ preserves positivity of the vector $z_{n}(x,\tau)$ and is convergent in the spectral norm as all $|\lambda_i|_n < 1, \ \forall i \in [1,N]$. 

%%%%%%%%%%%%%%%%%%%%%%%%%%%%%%%%%%%%%%%%%%%%%%%%%%%%%%%%%%%%
\section{Proof of Proposition~\protect{\ref{propMeixner1}}}
%%%%%%%%%%%%%%%%%%%%%%%%%%%%%%%%%%%%%%%%%%%%%%%%%%%%%%%%%%%%
The proof is analogous to the previous Proposition with allowance for the fact that matrix $2 a b A_2^B$ with $\beta < 0$ is the transpose of the matrix $2 a b A_2^F$ with $\beta \ge 0$.

\end{document}